\documentclass[journal]{IEEEtran}
\usepackage[T1]{fontenc}
\usepackage{mathptmx}
\usepackage{amsmath,amssymb,amsthm}
\usepackage{graphicx,booktabs,multirow,xcolor,tikz,url}
\usetikzlibrary{positioning,arrows.meta,fit,shapes.geometric}
\usepackage[hidelinks]{hyperref}
\hypersetup{pdftitle={Type-Safe Decision Frameworks for Agentic 5G Control},pdfauthor={Michail-Alexandros Kourtis, George Xilouris}}
\newtheorem{theorem}{Theorem}
\newtheorem{proposition}{Proposition}
\newtheorem{lemma}{Lemma}
\newtheorem{corollary}{Corollary}
\theoremstyle{definition}
\newtheorem{definition}{Definition}
\newcommand{\ACT}{\textsc{act}}
\newcommand{\ESC}{\textsc{esc}}
\newcommand{\ABS}{\textsc{abs}}
\newcommand{\Prob}{\mathbb{P}}

\begin{document}
\title{Type-Safe Decision Frameworks for Agentic 5G Control: A Theory-Driven Testbed Characterization of Where They Can Be Applied}
\author{Michail-Alexandros~Kourtis and George~Xilouris%
\thanks{M.-A. Kourtis and G. Xilouris are with the National Centre for Scientific Research ``Demokritos'' (NCSRD),
Athens, Greece.}}
\maketitle

\begin{abstract}
This paper presents a theory-driven characterization of type-safe decision frameworks for the agentic control of 5G
networks, where every decision must be an element of a declared option set rather than free text. Three design points are
evaluated on an Open5GS/UERANSIM testbed with a closed core-policy loop, namely a hosted typed model (Jev), an open
fine-tunable typed encoder (Laya), and a zero-label retrofit of a general language model (AnyJev). The proposed
theoretical framework turns timeliness, type conformance, certification cost and cardinality into checkable
applicability predicates, supported by an optimal act/escalate/abstain gate, an escalation-feasibility floor, a
co-location stability condition, per-type conformal risk control with a certification label floor, and a
type-mismatch bound. Measuring every predicate yields an applicability map from framework to 5G decision
class. Type safety removes format failures but not the question: the fine-tuned typed encoder returned its training
answer for 98--99.5\% of changed questions, and its calibrated gate then acted wrongly on up to
80\% of them, whereas the question-reading frameworks acted wrongly on at most
0.143 (Jev) and 0.137 (AnyJev) of any changed question, but were either hosted and
11--29 times slower (Jev) or reliant on an 8B language model (AnyJev).
\end{abstract}
\begin{IEEEkeywords}
Agentic AI, type safety, conformal risk control, O-RAN, 5G core network
\end{IEEEkeywords}

\section{Introduction}
\IEEEPARstart{T}{he} convergence of agentic Artificial Intelligence (AI) with the programmable control surfaces of 5G and
6G networks, namely the RAN Intelligent Controllers (RICs) of the Open Radio Access Network (O-RAN) and the service-based
5G Core (5GC), has created the need for new requirements in the design and deployment of network control loops in terms
of correctness, timeliness, auditability and energy efficiency~\cite{survey-agentic,agentic-airan}. An agentic controller turns
observations into actions: it selects a handover target, decides whether a slice's Service-Level Agreement (SLA) is
breached, grades an anomaly, or chooses which slice holds a priority flow. The network, however, can only execute an
action that belongs to the set of options it exposed, which imposes a stringent contract on the controller's output.

Large Language Models (LLMs) answer in free text, and therefore fail this contract by format, by parsing, or by answering
a different question than the one asked. Constraining their decoding to a grammar restores the output domain but
reduces accuracy, since it removes the free-text reasoning pass on which the model relies~\cite{speakfreely}, even when the constraint itself is enforced efficiently~\cite{outlines}.
Recent agentic O-RAN designs address the same concern at the system level: A1gent compiles operator goals into typed A1
policy instances enforced by deterministic near-real-time (near-RT) applications~\cite{a1gent}, the Contract-based
Agentic Intent Framework (CAIF) audits intents against formal RAN constraints before actuation~\cite{caif}, while
xTRUCE and AURA arbitrate conflicting agent proposals with safety and stability guarantees~\cite{xtruce,aura}. These designs govern actions \emph{after} an agent has proposed them; they do not characterize the decision function
that produces each proposal, which is where type safety, calibration and latency are decided.

In this respect, a new class of \emph{type-safe decision frameworks} has emerged, which returns, by construction, a
probability distribution over the declared options instead of text. Jev is offered by TypeSafe AI
as a hosted ``System One'' model with a stated end-to-end response time of 70--500\,ms~\cite{jev}; Laya is an open
421M-parameter typed encoder built on ModernBERT-large and trained with strictly proper scoring rules~\cite{laya,modernbert};
and AnyJev is a layer from Nokia Applied Research that reads an open LLM as a typed decider, without training at its
zero-label level, through the
probabilities of the option letters~\cite{anyjev}. Each framework guarantees the output type, yet none has been
characterized against the constraints that 5G control actually imposes: the deadlines of the O-RAN control
loops~\cite{oran-wg1}, questions whose type is declared only at runtime, the number of labels needed before an action
can be trusted, and the co-location of several models on a single edge Graphics Processing Unit (GPU).

This paper presents a theory-driven characterization of the three frameworks, which aims to answer what each one
offers, what each one costs, and \emph{where each one can be applied} in a 5G control stack. The contribution of the
presented work is four-fold: (i) a theoretical framework with eight evaluation axes, whose results compose into
applicability predicates for timeliness, type, coverage and cardinality, and which contributes a co-location stability
condition, a certification label floor, and a type-mismatch bound that is attained by question-blind models; (ii) an
experimental characterization of all three frameworks on an Open5GS/UERANSIM platform, which integrates a core-policy
loop from the Network Data Analytics Function (NWDAF) through the Policy Control Function (PCF) and the Session
Management Function (SMF) to the User Plane Function (UPF), with every endpoint pre-registered and reported whether it held or failed; (iii) an applicability
map from framework to 5G decision class, where every cell is either a measurement or a named violated predicate; and
(iv) a set of findings that the map makes precise, specifically that type safety removes format failures but not the
question, that the design points form a trade-off rather than a ranking, that a zero-shot framework is not zero-label
for certified action, that free-text escalation to an 8B-class on-edge LLM is infeasible at or below 1\,s on the
evaluated edge GPU, and that a co-located LLM breaks a 20\,Hz batch-1 decision loop at any escalation rate.

The paper is organized as follows. Section~\ref{sec:related} presents the type-safe frameworks and the related
technologies. Section~\ref{sec:model} defines the system model, and Section~\ref{sec:theory} presents the theoretical
framework. Section~\ref{sec:testbed} describes the evaluation architecture and methodology, Section~\ref{sec:results}
reports the characterization per axis, and Section~\ref{sec:apply} derives where each framework applies. Section~\ref{sec:discussion} discusses the
architectural, methodological and deployment implications, Section~\ref{sec:limits} states the limitations, and finally
Section~\ref{sec:concl} concludes the paper.

\section{Type-Safe Decision Frameworks and Related Technologies}
\label{sec:related}

\subsection{Type-safe decision frameworks}
Jev is TypeSafe AI's first ``System One'' model, described by its vendor as a function call that maps unstructured state
to typed probabilistic decisions; it samples in parallel instead of generating tokens sequentially, it is trained with a
method termed Reinforcement Learning for Calibrated Decisions (RLCD), and it is stated to answer in 70--500\,ms end to
end, 40--200$\times$ faster than LLMs on comparable tasks~\cite{jev}. The model is proprietary and is accessed through a
hosted Application Programming Interface (API), so the evaluated version (jev-1.13.0) is fixed only by its
identifier. Laya is a non-autoregressive decision model with 421M parameters, composed of a ModernBERT-large
backbone~\cite{modernbert} and a decision head with an option-marker scorer and an act/escalate output; it is trained with
reinforcement learning under strictly proper scoring rules, it is released under Apache-2.0, and it is stated to answer
in about 33\,ms per forward pass~\cite{laya}. AnyJev turns an open LLM into a typed decider: its zero-label level L0 removes the position bias of the option letters by averaging over cyclic option orders and divides out a label-free
prior, L1 adds temperature scaling on 100--500 labels, and L2 fits a closed-form head on a mid-depth hidden state from
100--300 labels per question layout, with heads shipped for the Qwen3 family~\cite{anyjev,qwen3}. The three frameworks
therefore span three distinct design points, namely a hosted specialist, an open fine-tunable specialist and an open
retrofit, and each is type-safe by construction. None of them, however, states under which deadlines, question types
and label budgets its guarantees hold in a network.

\subsection{Agentic control in O-RAN and the 5G Core}
The agentic view of the RAN organizes the task landscape around slice life-cycle management, Radio Resource Management
(RRM) closed loops and cross-cutting security, and introduces planning, tool use and self-management gates into
long-lived control loops~\cite{agentic-airan}. A1gent separates non-real-time (non-RT) agentic reasoning from near-RT deterministic
execution through typed A1 policy instances and a fixed-priority action merger~\cite{a1gent}; CAIF decouples
probabilistic intent extraction from strictly governed policy execution and evaluates the design on an O-RAN platform
for network slicing~\cite{caif}. When several agents act on shared resources, xTRUCE places a provably safe arbiter in the
near-RT RIC with a three-layer constraint hierarchy~\cite{xtruce}, and AURA shows on a live O-RAN system that two agents
with individually correct objectives drive recurring opposing excursions, which an arbitration layer with feasibility
invariants and dwell times removes~\cite{aura}. OpenTwin certifies a digital twin by re-simulation and evaluates each
xApp action before it executes, after observing that an E2 control request may be reported as successful while the base
station never applies it~\cite{opentwin}. A recent tutorial maps agentic capabilities onto the 5G and 6G control planes
and standardization~\cite{survey-agentic}. These works secure the \emph{composition} of agent actions; the presented work
is complementary and characterizes the \emph{decision function} that produces each action, which is where type safety,
calibration and latency are decided.

\subsection{Certified abstention and conformal methods in wireless networks}
Chow's reject rule trades error against the rate of refusal~\cite{chow}, and selective classification extends the trade
to deep networks~\cite{geifman}. Split conformal prediction~\cite{vovk} and Conformal Risk Control (CRC)~\cite{crc} make
such trades valid in finite samples under exchangeability, and role-stratified CRC assigns a separate threshold and risk
budget to each argument role of an LLM tool call, so that a rare high-risk field is not averaged away~\cite{rolecrc}. In
wireless networks, meta-learned context-dependent conformal prediction calibrates O-RAN applications when calibration and
runtime contexts differ~\cite{metacp}, post-hoc conformal prediction quantifies, after the fact, the miscoverage of prediction sets whose size is fixed by an
operational constraint rather than by a prescribed level~\cite{posthoccp}, and confounding-valid conformal inference bounds counterfactual Key Performance Indicators (KPIs) from logged telemetry complemented by scarce randomized
telemetry~\cite{cfkpi}. Cascades that defer to a stronger model on low confidence expose a further
attack surface, since an adversary can lower the weak model's confidence and force deferral~\cite{forceddeferral}. The
presented work applies CRC to the joint event \{\ACT{} and wrong\} per decision type, and shows why its guarantee must be
indexed by the declared question type and not only by the output domain.

\subsection{Control timescales, serving and emulation}
O-RAN places near-RT control loops between 10\,ms and 1\,s and non-RT loops above 1\,s~\cite{oran-wg1}, while
5GC analytics through the NWDAF~\cite{3gpp-23288} and policy authorization over the N5 interface~\cite{3gpp-29514}
operate at the core-policy timescale. Stochastic network calculus has been used to bound the probability that the delay of Ultra-Reliable Low-Latency
Communications (URLLC) services exceeds its budget in the near-RT loop~\cite{oranus}. An LLM acting as a
second-stage decider is typically served with batched decoding~\cite{vllm} and quantized weights~\cite{awq}, which fixes
its time per output token and couples its latency to the load. Finally, open-source 5G platforms expose similar
interfaces under different timing fidelity, so that functional compatibility does not imply timing
fidelity~\cite{atlasran}; the presented platform therefore scopes its RAN-side timing by a real-time factor.

\section{System Model and Type Safety}
\label{sec:model}
The proposed system model separates the declared decision from the framework that answers it, in order to express every
property of the characterization as a property of a (framework, decision class) pair.
\begin{definition}[Decision class]
A decision class $d$ is a question template $t_d$ with option schema $\mathcal{O}_d$, $|\mathcal{O}_d| = K_d$, a deadline
$D_d$, a deadline-miss budget $\varepsilon_d$, a risk budget $\alpha_d$, a coverage floor $c_d$, a network latency
$L_{\mathrm{net},d}$ and a host.
A \emph{decision type} is the pair $(t_d, \mathcal{O}_d)$.
\end{definition}
\begin{definition}[Type-safe framework]\label{def:ts}
A System-1 $F$ is type-safe for $d$ if, for every state $s$, it returns a distribution $p_F(\cdot\mid s, t_d)$ on
$\mathcal{O}_d$; its action is $\arg\max p_F$ and its confidence $m = \max p_F$.
\end{definition}
A gate maps $(m, \sigma)$, with $\sigma$ the slack left after System-1, to one of three actions: (i) \ACT{}, which executes
System-1's option; (ii) \ESC{}, which escalates the decision to a System-2 LLM; and (iii) \ABS{}, which applies the
class's safe default. The loss of the default is $\ell_0$, and $\beta := 1-\ell_0$ is the accuracy an action needs in
order to improve on it. A generative model can also be made type-safe, by masking at every decoding step the tokens that
do not continue some option; a type-safe framework instead scores the options directly and returns a law over
$\mathcal{O}_d$ without any decoding.

\section{Theoretical Framework}
\label{sec:theory}
The proposed framework organizes the characterization into eight axes, listed in Table~\ref{tab:axes}; the axes that carry
the main results are backed by a formal result and measured in Section~\ref{sec:results}, and the proofs are given in the
Appendix.

\begin{table}[t]\centering\footnotesize
\caption{Evaluation axes, the formal result behind each, and where it is measured.}\label{tab:axes}
\begin{tabular}{@{}p{0.23\columnwidth}p{0.36\columnwidth}p{0.3\columnwidth}@{}}
\toprule
Axis & Formal object & Measured by \\ \midrule
A0 type safety & output $\in\mathcal{O}_d$ (Def.~\ref{def:ts}) & typed interface \\
A1 decision quality & proper-score decomposition & known-type accuracy \\
A2 certified action & per-type CRC gate (Thm.~\ref{thm:crc}) & Table~\ref{tab:gate} \\
A3 label floor & Lemma~\ref{lem:C} & label curves (Table~\ref{tab:labels}) \\
A4 type conformance & Prop.~\ref{prop:ts} & Fig.~\ref{fig:ts} \\
A5 timeliness & Thm.~\ref{thm:gate}, Lem.~\ref{lem:floor}, Prop.~\ref{prop:coloc} & Figs.~\ref{fig:time}, \ref{fig:r1} \\
A6 cardinality & option budget $K_F^{\max}$ & handover $k\le64$ (Fig.~\ref{fig:map}) \\
A7 rewording & robustness vs.\ insensitivity & rewordings; counterfactual questions \\
A8 cost, sovereignty & J/decision; data leaves the domain & L4 energy; hosting \\ \bottomrule
\end{tabular}
\end{table}

\subsection{The optimal gate and the near-RT collapse (A5)}
Let System-1 be calibrated within the class, $\Prob(\text{correct}\mid m)=m$. Relative to the default, $V_{\ACT}(m)=m-\beta$,
$V_{\ABS}=0$ and $V_{\ESC}(m,\sigma)=\mathbb{E}[\mathbf{1}\{L_2\le\sigma\}(C_2-\beta)\mid m]-\mu\Prob(L_2>\sigma\mid m)-\kappa$,
where $C_2$ is System-2's correctness, $L_2$ its latency, $\mu\ge0$ the multiplier on deadline misses, $\kappa\ge0$ the
escalation cost, and $F_2(\sigma):=\Prob(L_2\le\sigma\mid m)$ the conditional latency distribution, which under the
independence assumed below does not depend on $m$; the argument $\sigma$ is omitted where clear. These values arise from
the following program: the gate minimizes the expected escalation cost $\mathbb{E}[c(a)]$ over randomized rules
$(m,\sigma)\mapsto a\in\{\ACT,\ESC,\ABS\}$, subject to a per-class risk constraint $\mathbb{E}[\ell\mid d]\le\alpha_d$ on
the loss $\ell\in[0,1]$ of the executed option and a per-class deadline-miss constraint $\Prob(\text{miss}\mid d)\le\varepsilon_d$.
With multipliers $\lambda,\mu\ge0$ on the two constraints, dividing the Lagrangian by $\lambda$ expresses every term in
loss units, so that $\kappa=c_2/\lambda$ is the escalation cost and $\mu$ the rescaled price of a missed deadline.
\begin{theorem}[Threshold gate]\label{thm:gate}
For fixed multipliers, the optimal gate picks $\arg\max\{V_{\ACT},V_{\ESC},V_{\ABS}\}$ pointwise. If $L_2$ is independent of
$(C_2,m)$ and $\pi_2(m):=\Prob(C_2=1\mid m)\le m$ for all $m$, the \ACT{} region is exactly $\{m\ge\beta\}$ for every slack.
\end{theorem}
\begin{corollary}[Collapse to act-or-abstain]\label{cor:collapse}
The \ESC{} region is empty at slack $\sigma$ if $F_2(\sigma)=0$, i.e., System-2 cannot answer in time, or if $L_2$ is
independent of $(C_2,m)$ and $\pi_2(m)\le\max(m,\beta)$ for all $m$. The optimal gate is then Chow's rule with the safe default~\cite{chow}: \ACT{} iff $m\ge\tau$,
else \ABS.
\end{corollary}
It is worth noting that the corollary turns the escalation question into two measurable conditions, namely whether
System-2 can answer within the slack and whether it is more accurate than System-1 where System-1 is unsure.
Two remarks qualify the theorem in practice. First, the independence of $L_2$ from $(C_2,m)$ is an idealization: on the
presented platform the System-2 answers that arrive early are the short ones, and short answers are more often correct,
so that the value of escalating is not monotone in the slack unless a missed deadline is priced explicitly. The
multiplier $\mu$ therefore has to reflect how the operator counts a miss, and a replay with $\mu=0$ escalated
decisions whose answers then expired (Section~\ref{sec:results}). Second, the theorem is stated per decision class,
and the classes in which $\pi_2>m$ holds at all are an empirical property of the System-1/System-2 pair; the
characterization measures them rather than assuming them.

\subsection{Escalation feasibility and co-location (A5)}
\begin{lemma}[Service-time floor]\label{lem:floor}
A System-2 answer of $n$ tokens takes at least $\mathrm{TTFT}+(n-1)\cdot\mathrm{TPOT}$, with TTFT the time to the first
token and TPOT the per-token decode time at concurrency 1, which is the shortest the server attains since batching only
lengthens each decode step. Escalation is feasible within $D$ only for answers with $n\le1+(D-\mathrm{TTFT})/\mathrm{TPOT}$.
\end{lemma}
\begin{proposition}[Co-location stability]\label{prop:coloc}
Let System-1 decisions arrive every $T$ and be served in order, taking $s_i$ while a co-located System-2 is idle and $s_b$
while it decodes ($s_i<T<s_b$), and let System-2 be busy a fraction $u_2$ of the time, alternating slowly relative to $s_b$.
System-1 is stable iff $u_2<u^*:=s_b(T-s_i)/\big(T(s_b-s_i)\big)$, and within a System-2 busy period of length $B$ its
lateness grows to $B(s_b-T)/T$.
\end{proposition}
The proposition makes explicit a cost that a per-request latency figure hides: a System-1 that meets its deadline in
isolation can still fall behind its period whenever it shares the GPU with a decoding System-2.

\subsection{Certified action per type (A2, A3, A4)}
Let $Z=(m,\mathbf{1}\{\text{wrong}\})$ and calibrate on $n$ labelled decisions of type $t$, restricting the threshold to the
candidate set $T_n=\{m_1,\dots,m_n,\infty\}$:
$\hat\tau=\min\{\tau\in T_n:(\sum_i\mathbf{1}\{m_i\ge\tau,\text{wrong}_i\}+1)/(n+1)\le\alpha\}$, with $\hat\tau=\infty$ if no
finite threshold qualifies.
\begin{theorem}[Per-type gate~\cite{crc}]\label{thm:crc}
If calibration and test decisions of type $t$ are exchangeable, $\Prob_t(\ACT\wedge\text{wrong})\le\alpha$. For i.i.d.\
calibration from $P$ and an independent test decision drawn from $Q$, the risk is at most $\alpha+d_{\mathrm{TV}}(P_Z,Q_Z)$.
\end{theorem}
\begin{lemma}[Certification label floor]\label{lem:C}
The gate can act ($\hat\tau<\infty$) only if $n\ge\lceil 1/\alpha\rceil-1$, i.e., 19 labels at $\alpha=0.05$
and 9 at $\alpha=0.10$, per decision type.
\end{lemma}
Zero-shot frameworks therefore pay only this floor for a new type, while specialists pay their training labels on top:
zero-shot is not zero-label for certified action.
The floor is a necessary condition only: the \ACT{} rate that the gate reaches at a given $\alpha$ grows with $n$,
since the finite-sample term $1/(n+1)$ forces a conservative threshold on small calibration sets, and a framework whose
confidence separates correct from wrong decisions poorly may still act on almost nothing after calibration. The
label cost of a new decision type is thus the sum of three terms, namely the training labels of a specialist, the
calibration floor of Lemma~\ref{lem:C}, and the additional calibration labels that bring the \ACT{} rate above the
coverage floor $c_d$.
\begin{proposition}[Type mismatch]\label{prop:ts}
If the gate of type $t$ answers a query of type $t'\neq t$, then
$\Prob_{t'}(\ACT\wedge\text{wrong})\le\alpha_t+d_{\mathrm{TV}}(P^t_Z,P^{t'}_Z)$, and no better bound holds in general: a
System-1 that ignores the question has the same law of $m$ under $t$ and $t'$ while its correctness can flip, so the
risk approaches the \ACT{} rate.
\end{proposition}
For such a question-blind System-1 no test on $m$ can detect the mismatch; only a check of the declared type can. The
guarantee of Theorem~\ref{thm:crc} is thus a property of the (framework, type) pair, not of the framework.

\subsection{Cardinality (A6)}
Let $K_F^{\max}$ be the largest option count that $F$ scores in one call. A decision over $K>K_F^{\max}$ options can be
served by a round of groups of $g\le K_F^{\max}$ options followed by a final among the $\lceil K/g\rceil$ group winners, at
the price of more than one call per decision and of an accuracy that is the product of the round accuracies; the
analysis of this reduction is reported in an extended version of this work, and the present paper uses only its measured
accuracy and latency where $K_d>K_F^{\max}$.

\subsection{Applicability predicates}
\begin{proposition}[Applicability]\label{prop:A}
Under the exchangeability of Theorem~\ref{thm:crc}, $F$ serves class $d$ with $\Prob(\ACT\wedge\text{wrong})\le\alpha_d$,
$\Prob(\text{miss})\le\varepsilon_d$ and \ACT{} rate $\ge c_d$ if
(T) $q_{1-\varepsilon_d}(L_{\mathrm{net},d}+L_F(K_d))\le D_d$, where $L_F(K_d)$ is $F$'s per-decision latency at $K_d$ options
and $q_{1-\varepsilon}$ the $(1-\varepsilon)$-quantile, with the latency of the group-and-final reduction when $K_d>K_F^{\max}$ and
Proposition~\ref{prop:coloc}'s condition on a shared GPU;
(Y) $t_d$ is in $F$'s calibrated domain with at least $\lceil 1/\alpha_d\rceil-1$ labels;
(C) the gate's \ACT{} rate at $\alpha_d$ is at least $c_d$; and
(K) $K_d\le K_F^{\max}$, directly or through the group-and-final reduction.
(Y) is necessary in the worst case (Prop.~\ref{prop:ts}) and (T) by definition.
\end{proposition}
The proposition is shallow by design: it turns the question of where $F$ can be applied into a predicate that a
measurement settles, which is the purpose of the characterization that follows.

\section{Evaluation Architecture and Methodology}
\label{sec:testbed}
The presented evaluation architecture is designed to provide a controlled and reproducible 5G environment, integrating
an open-source 5G Core, an emulated RAN, a typed decision service and a GPU-hosted System-2. The system is structured
into multiple layers in order to ensure that every framework is evaluated behind the same typed interface, on the same
decisions and under the same timing conditions, as depicted in Fig.~\ref{fig:testbed}. The primary ambition of the
architecture is to attribute every difference in the results to the framework, and not to the path around it. We implemented every layer with open-source components, except for the hosted framework, as detailed below.

\begin{figure}[t]\centering
\begin{tikzpicture}[font=\scriptsize, node distance=4mm and 7mm,
  box/.style={draw, rounded corners=1pt, align=center, minimum height=5.5mm, fill=white},
  arr/.style={-{Stealth[length=1.6mm]}, thick}]
\node[box] (ue) {UEs, gNB\\(UERANSIM)};
\node[box, right=of ue] (upf) {per-slice\\UPF};
\node[box, right=of upf] (smf) {per-slice\\SMF};
\node[box, right=of smf] (pcf) {PCF};
\node[box, below=9mm of smf] (nwdaf) {NWDAF};
\node[box, below=9mm of pcf, fill=blue!8] (agent) {loop\\agent};
\node[box, below=6mm of agent, fill=orange!15] (s1) {System-1 service:\\Jev $|$ Laya $|$ AnyJev\\(typed \texttt{/v1/judge})};
\node[box, left=of s1, fill=gray!12] (gate) {gate:\\act / abstain\\(escalate)};
\node[box, left=of gate, fill=gray!12] (s2) {System-2\\Qwen3-8B\\(vLLM)};
\draw[arr] (ue) -- (upf); \draw[arr] (smf) -- node[below, font=\tiny]{PFCP} (upf); \draw[arr] (pcf) -- node[below, font=\tiny]{N7} (smf);
\draw[arr] (nwdaf) -- node[below]{notify} (agent); \draw[arr] (agent) -- node[right]{N5} (pcf);
\draw[arr] (agent) -- node[right]{choice} (s1); \draw[arr] (s1) -- (gate); \draw[arr, dashed] (gate) -- (s2);
\node[draw, dashed, rounded corners, fit=(ue)(upf)(smf)(pcf)(nwdaf), inner sep=2mm, label={[font=\scriptsize]above:Open5GS core, Plane B (built)}] {};
\end{tikzpicture}
\caption{The presented evaluation architecture and the framework plug-in point. Each framework serves the same typed
endpoint, and the core-policy loop (Plane B) runs end to end; Plane A (E2/near-RT RIC) is not closed on this
platform.}\label{fig:testbed}
\end{figure}
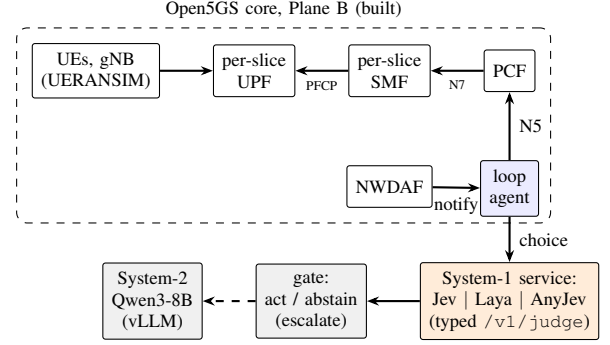

\subsection{Architecture layers}
\textbf{Core network layer.} This layer is realized by Open5GS~2.8.0~\cite{open5gs} with a Session Management Function
(SMF) and UPF pair per slice, namely enhanced Mobile Broadband (eMBB), URLLC and massive Machine-Type Communications
(mMTC). User Equipment (UE) and the gNodeB are emulated by UERANSIM~\cite{ueransim} with two local patches, specifically
the PDU session resource modify procedure and uplink QoS-flow classification, so that a policy change reaches the UPF data
path in both directions (UERANSIM commit 48554b7 with patches 0001 and 0002). An srsRAN~\cite{srsran} gNodeB over ZeroMQ
is available on the platform, but no RAN-side timing result is reported in this paper. There is no real radio; RAN-side
timing is scoped by a real-time factor~\cite{atlasran} and lies outside the presented results. The decision host and the
core run on different hypervisor nodes, and since the platform's switch does not trunk the dedicated decision-path VLANs
between them, the loop agent's traffic to the System-1 endpoint and to the PCF traverses the shared 1\,GbE management VLAN. This affects only transport latency, not the signalling path: every policy is
installed through N5 at the PCF.

\textbf{Analytics and loop layer.} An Open5GS-targeted NWDAF~\cite{nwdaf-impl} observes logs and counters and notifies a
loop agent, which turns each notification into a typed three-slice choice, sends it to the decision layer, and installs
the chosen policy through N5 policy authorization at the Policy Control Function (PCF), from where it propagates to the SMF
over N7 and to the UPF over the Packet Forwarding Control Protocol (PFCP). No 3GPP conformance of the NWDAF is claimed.

\textbf{Typed decision layer.} This layer exposes one typed endpoint (\texttt{/v1/judge}) behind which each framework is
plugged in turn. Laya is used zero-shot (``base'') and after full fine-tuning on the v2 training split (5
seeds); a plain encoder, namely ModernBERT-base~\cite{modernbert} with one linear head per fixed-cardinality class, is trained on the same split (5 seeds) as the task-specific reference. AnyJev reads Qwen3-8B (bf16)~\cite{qwen3} zero-shot through
cyclic-order marginalization of the option-letter softmax with a label-free batch prior (L0); its L2 variant fits a
closed-form head per question layout on the training split, which the presented work pools per option count
($\dagger$, an extension to the shipped levels). Jev (jev-1.13.0) is queried zero-shot over the Internet.
Software versions: Laya 0.3.5 and AnyJev commit 3cd8c6f on PyTorch 2.6.0 (CUDA 12.4); Qwen3-8B (bf16) and Qwen3-8B-AWQ
on vLLM 0.30.0. Every gate in this paper operates on each framework's top-option probability $\max p_F$, not on the confidence
field that some frameworks report next to it.

\textbf{Gate and escalation layer.} The gate implements the conformal risk-control threshold of Theorem~\ref{thm:crc} on
each framework's top-option probability, fitted per decision class on the held-out calibration partition (Laya, and the
encoder on its direct classes) or, where a framework has no calibration run, by 2-fold cross-fitting on the test split
(Jev, AnyJev, and the encoder's cells at $k\ge16$), i.e., each fold's threshold is fitted on the other half of
the test items; every such choice is logged as a pre-registration deviation. Cross-fitting halves the effective
calibration size, so at $k\ge16$ the finite-sample term $1/(n+1)$ of Theorem~\ref{thm:crc} is about 0.015, i.e., 30\% of
$\alpha=0.05$, and the corresponding cells of Table~\ref{tab:gate} should be read with that granularity. The escalation target is Qwen3-8B served by vLLM~\cite{vllm} with 4-bit Activation-aware Weight Quantization
(AWQ)~\cite{awq}, co-located with System-1 on the edge GPU.

\textbf{Measurement layer.} We measured System-1 latency on an NVIDIA L4 (the edge GPU) with the Streaming Multiprocessor
(SM) clock locked at 1050\,MHz after a 20\,s settle, at batch 1 from raw per-call samples, reported as p50 and p99, and energy per decision is computed as mean power times
mean latency from 20\,Hz power samples; we ran training and accuracy runs on an NVIDIA RTX~4090 whose timings are never
reported. Jev's latency includes the wide-area network and is reported for cold (a new TLS session per call) and warm
connections, while its energy cannot be measured locally. The co-location and catch-up-batching runs
(Section~\ref{sec:results}) serve Laya's base checkpoint at $k=8$ and 20 decisions/s, since latency does not depend on the
weights, against Qwen3-8B-AWQ under vLLM on the same locked L4. The trace-driven gate replay uses fine-tuned Laya with
isotonic-calibrated confidence as System-1, System-2 correctness from the bf16 Qwen3-8B on the RTX~4090 and System-2
latency from the AWQ deployment on the L4, with $\pi_2$ and $F_2$ 2-fold cross-fitted on the test split; the AWQ
deployment is less accurate on handover than the bf16 model, so bf16 correctness biases the replay toward escalation, and
the set of escalated decisions was unchanged when AWQ correctness was substituted.

\subsection{Decisions, data and protocol}
The frozen v2 corpus comprises 3,645 test items over eight known decision classes, namely handover target selection with
$k\in\{4,\dots,64\}$ neighbours, SLA breach, urgency and anomaly severity, with labels generated from simulator rules.
Per class, the test partition holds 540 items for handover $k=4$ and $k=8$, urgency and anomaly severity,
135 items for each of handover $k=16$, $32$ and $64$, and 1,080 items for SLA breach; the
calibration and shift partitions have the same per-class sizes. Each item also carries a slice criticality, high (URLLC),
medium (eMBB) or low (mMTC), which assigns it a declared deadline of 100, 500 or 1,000\,ms in the gate replay, where the
class risk budget is scaled by 0.5, 1 and 2 respectively. The five seeds vary the training run of the fine-tuned arms and
share one calibration and test split.
Two derived corpora probe generality: (i) eight \emph{runtime-declared templates} (2,400 test items), which keep a v2
item's state and options and ask a different question; and (ii) the v2 test items under four \emph{rewordings}.
Each v2 item carries a network state, i.e., slice-level Key Performance Measurements (KPMs) such as the 95th-percentile
latency, the slice's SLA latency target and per-neighbour load and Reference Signal Received Power (RSRP), together with a
typed question of one of three judgment types: (i) \emph{choice}, which selects one of $k$ options, e.g., the handover
target among $k$ neighbour cells; (ii) \emph{yes/no}, which returns a probability, e.g., whether the slice's SLA is
breached; and (iii) \emph{score}, which selects an ordered level, e.g., the urgency of an intervention or the severity of
an anomaly. The corpus is split into training, calibration, test and shift partitions, and the v2 corpus has been frozen
since the first result. The runtime-declared templates ask, over the same states and options, among others for the most
loaded neighbour, for the least loaded neighbour subject to an RSRP constraint, for the neighbour with the strongest
signal, whether the jitter exceeds a fraction of the SLA, for a congestion level, for the root-cause KPI of an anomaly,
and for the known handover question under renamed option keys. Four of them change only the question while keeping the
state and options of a v2 item, which isolates whether a framework reads the question at all. The rewordings comprise
(i) the rule stated in words, (ii) the rule stated through the state's field names, and (iii) two paraphrases written
from a template rather than generated by any framework under evaluation. Shift
penalties use a sample-split (Scheff\'e-set) estimate of total variation on $Z$, reported next to a null estimate on two
halves of one sample. We pre-registered every endpoint before its run, with a non-inferiority margin of 3\,pts, paired
bootstrap over items (10,000 resamples) and 5 seeds wherever training is involved, and every departure from a
registration is logged. Byte-identical items reach every framework, and each cell states its label use. The applicability map evaluates
Proposition~\ref{prop:A} on the deadline grid $\{100\,\text{ms},500\,\text{ms},1\,\text{s},5\,\text{s},60\,\text{s}\}$,
which spans the near-RT, core-policy and non-RT timescales of Section~\ref{sec:related}, with $\varepsilon_d=0.01$ (so that
(T) uses the p99 latency), $\alpha_d=0.05$ and a coverage floor $c_d=0.5$, i.e., a framework counts as applicable only
if its certified gate acts on at least half of the decisions; these parameters were fixed in the pre-registration, and the
map is reported for a dedicated GPU, with the shared-GPU condition of Proposition~\ref{prop:coloc} treated separately in
Section~\ref{sec:results}.
Table~\ref{tab:fw} summarizes the design points with the measurements of Section~\ref{sec:results}.

\begin{table}[t]\centering\scriptsize
\caption{The design points, with measured values: known-type accuracy (mean over handover $k{=}4,8$, SLA, urgency,
anomaly; test), question-reading (mean accuracy on the four counterfactual templates), p99 latency and energy of one
$k{=}4$ decision on the L4 (Jev: end to end, warm), and labels needed before certified action on a new type.}\label{tab:fw}
\setlength{\tabcolsep}{2.2pt}
\resizebox{\columnwidth}{!}{%
\begin{tabular}{@{}lccccc@{}}\toprule
 & known acc. & question-reading & p99 (ms) & J/dec. & labels for a new type \\\midrule
Laya FT & 0.90 & 0.41 & 25.8 & 1.10 & training + 19 \\
Laya base & 0.27 & 0.27 & 25.8 & 1.10 & 19 (gate rarely acts) \\
Encoder & 0.92 & no readout & 23.5 & 0.55 & training (new head) + 19 \\
AnyJev-8B L0 & 0.44 & 0.81 & 283 & 19.6 & 19 \\
AnyJev-8B L2$\dagger$ & 0.73 & no readout & 94 & 6.6 & head per layout + 19 \\
Jev & 0.67 & 0.99 & 596 & n/m & 19; data leaves the domain \\\bottomrule
\end{tabular}}
\end{table}

\section{Characterization Results}
\label{sec:results}

\subsection{A0/A1: type safety and known-type quality}
Each framework returns a distribution over the declared options by construction, so format and parsing failures do not
arise at the typed interface, and every framework is compared on the same options. On known types, the fine-tuned
specialists reach specialist accuracy (Table~\ref{tab:fw}: encoder 0.92, Laya FT 0.90), while
the zero-shot frameworks remain substantially lower on the rule-defined classes (Jev 0.67, AnyJev L0
0.44, Laya base 0.27); the exception is SLA breach, which Jev answers at
0.963.

\subsection{A2/A3: certified action and its label cost}
\begin{table}[t]\centering\scriptsize
\caption{Certified gate per framework and known class ($\alpha=0.05$, test): $\Prob(\ACT\wedge\text{wrong})$ and \ACT{} rate.
Laya and the encoder's direct classes: 5-seed mean, calibrated on the held-out calibration partition; Jev, AnyJev and
the encoder's cells at $k\ge16$: 2-fold cross-fitting on the test items; $n=540$ per class
(135 at $k\ge16$, 1,080 for SLA breach). Bold: above $\alpha$, all within finite-sample noise
(one item is 0.001--0.007).}\label{tab:gate}
\setlength{\tabcolsep}{2.4pt}
\resizebox{\columnwidth}{!}{\begin{tabular}{lrrrrrrrr}
\toprule
 & \multicolumn{2}{c}{Laya FT} & \multicolumn{2}{c}{Encoder} & \multicolumn{2}{c}{AnyJev-8B L0} & \multicolumn{2}{c}{Jev} \\
\cmidrule(lr){2-3} \cmidrule(lr){4-5} \cmidrule(lr){6-7} \cmidrule(lr){8-9}
class & risk & ACT & risk & ACT & risk & ACT & risk & ACT \\
\midrule
handover $k{=}4$ & 0.043 & 0.91 & 0.039 & 0.99 & 0.042 & 0.29 & 0.032 & 0.35 \\
handover $k{=}8$ & 0.042 & 0.69 & 0.046 & 0.96 & 0.046 & 0.12 & 0.044 & 0.23 \\
handover $k{=}16$ & 0.024 & 0.76 & 0.035 & 0.86 & 0.022 & 0.09 & \textbf{0.070} & 0.34 \\
handover $k{=}32$ & 0.047 & 0.63 & \textbf{0.052} & 0.74 & 0.037 & 0.09 & 0.038 & 0.24 \\
handover $k{=}64$ & 0.040 & 0.43 & 0.042 & 0.70 & 0.037 & 0.07 & 0.037 & 0.22 \\
SLA breached (yes/no) & \textbf{0.052} & 0.99 & \textbf{0.057} & 0.99 & 0.044 & 0.05 & 0.037 & 0.78 \\
urgency (score) & \textbf{0.056} & 0.79 & \textbf{0.060} & 0.85 & 0.044 & 0.08 & 0.026 & 0.21 \\
anomaly (score) & \textbf{0.051} & 0.81 & 0.043 & 0.84 & 0.043 & 0.06 & 0.042 & 0.08 \\
\bottomrule
\end{tabular}
}
\end{table}
Table~\ref{tab:gate} shows that the risk promise holds for every framework in distribution; what differs is how often a
framework can act at that risk, which is the practical utility of the gate. The specialists act on 43--99\% of
decisions, zero-shot AnyJev on at most 29\% and Jev on 8--78\%. Under the shift split, AnyJev L0's SLA gate is the one cell that leaves its budget: its joint risk rises to
0.199, about four times $\alpha$, while its \ACT{} rate rises from 0.05 to 0.26; the
sample-split shift on that framework's own $Z$ is 0.35, so the cell stays within the bound
$\alpha+d_{\mathrm{TV}}$ of Theorem~\ref{thm:crc}, but the bound is loose there, and the cell shows that a zero-label
confidence can be shift-sensitive even where the state distribution barely moves. It should be noted that a usable shift penalty requires a
debiased estimate: the plug-in total variation on $Z$ read 0.09--0.16 at high $k$, whereas the
sample-split estimate read 0.007--0.034.

\begin{table}[t]\centering\footnotesize
\caption{Labels to certified action for a new family, held out of both arms' training (leave-one-family-out, LOFO): labels
$n$ at which the arm first comes within 3\,pts of its own full-data accuracy (5 seeds), plus the Lemma~\ref{lem:C}
floor.}\label{tab:labels}
\setlength{\tabcolsep}{4pt}
\begin{tabular}{@{}lrrr@{}}\toprule
family & Laya (LOFO+FT) & encoder & ratio \\\midrule
SLA breached & 100 & 300 & 3.00 \\
urgency & 300 & 1000 & 3.33 \\
handover & 1000 & $>$1000 & $>$ 1.00 \\
anomaly & 300 & 100 & 0.33 \\\midrule
\multicolumn{4}{@{}l}{zero-shot (Jev, AnyJev L0): 19 labels/type at $\alpha=0.05$}\\\bottomrule
\end{tabular}
\end{table}
Table~\ref{tab:labels} presents the label axis. Both arms follow the same leave-one-family-out (LOFO) protocol: each is
first trained on the other three v2 families, and then continued on $n$ labels of the held-out family with the whole
model updated, under the same small-budget schedule; the arms differ in their backbone, ModernBERT-large (421M) for Laya
against ModernBERT-base for the encoder, which is a confound that this comparison does not separate. The pre-registered
claim that the typed encoder needs three times fewer labels than a plain encoder held in 2 of 4 families: Laya
is far more label-efficient on handover at 100--300 labels (0.726/0.828 against
0.225/0.321), but reaches its own parity only at 1,000, and on anomaly severity
the plain encoder converges faster. At zero
labels the held-out family is not served (e.g., handover 0.141). A \emph{new question type}, however, is
learnable: with 100 labels per runtime-declared template, continued full fine-tuning lifts Laya to
0.78--0.97 on all eight templates (``most loaded'' 0.009 $\rightarrow$
0.809), whereas 10 labels do not suffice (0.14--0.76) and neither does
training the head alone (0.123 on ``most loaded'' at 100). It can be deduced that the specialist must
update its encoder, and not only its head, in order to read a new question.

\subsection{A4: type conformance}
\begin{figure*}[t]\centering\includegraphics[width=\textwidth]{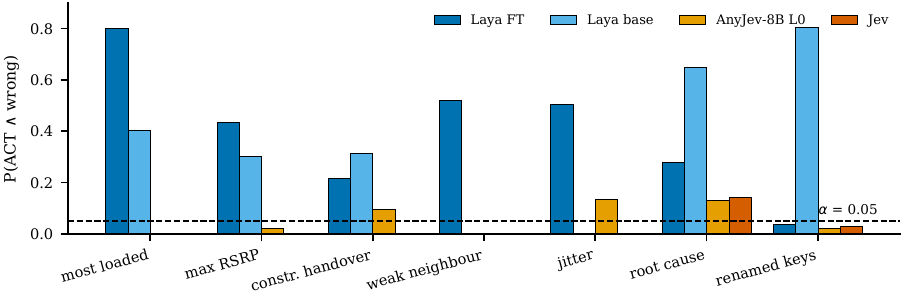}
\caption{Type mismatch: $\Prob(\ACT\wedge\text{wrong})$ when each framework's known-class gate ($\alpha=0.05$) answers the
seven gated runtime-declared templates (``max RSRP'' is the strongest-signal template and ``constr.\ handover'' the
RSRP-constrained one of Table~\ref{tab:g3}). Dashed: the $\alpha$ each gate promises on its own class.}\label{fig:ts}
\end{figure*}
\begin{table}[t]\centering\scriptsize
\caption{Question-reading on the runtime-declared templates (accuracy, test, zero labels; Laya FT: 5-seed mean). Last
column: share of item-seed pairs on which Laya FT returned the answer it gives to the source v2 question.}\label{tab:g3}
\setlength{\tabcolsep}{2.2pt}
\begin{tabular}{@{}lccccc@{}}\toprule
template & Laya base & Laya FT & AnyJev L0 & Jev & same answer \\\midrule
most loaded neighbour & 0.197 & 0.009 & 0.990 & 0.987 & 0.993 \\
least loaded, RSRP $\ge$ bound & 0.143 & 0.691 & 0.713 & 0.987 & 0.991 \\
strongest signal & 0.153 & 0.437 & 0.927 & 1.000 & 0.995 \\
jitter above SLA share & 0.603 & 0.493 & 0.597 & 1.000 & 0.982 \\
weak neighbour & 0.570 & 0.461 & 0.743 & 1.000 & -- \\
congestion level & 0.287 & 0.351 & 0.573 & 1.000 & -- \\
root-cause KPI & 0.343 & 0.329 & 0.167 & 0.803 & -- \\
renamed keys & 0.193 & 0.825 & 0.690 & 0.677 & -- \\\bottomrule
\end{tabular}
\end{table}
On templates that keep a known item's state and options and change only the question, fine-tuned Laya returned its
known-question answer for 0.982--0.995 of item-seed pairs; asked for the \emph{most} loaded neighbour, it
scored 0.009. Its calibrated gate kept acting (0.808 of those items), and
0.801 of all items ended as act-and-wrong, as depicted in Fig.~\ref{fig:ts}. The pre-registered
retention endpoint nonetheless passed ($+$0.138 [$+$0.114, $+$0.163]), since templates whose answer correlates with
the trained rule dominate the pool; we report the endpoint as registered, with this measurement beside it. In contrast,
Jev read the questions (act-and-wrong $\le 0.032$ on 6 of 7 templates), and AnyJev L0 did so on most
templates. The mismatch evaluation covers the seven templates whose answer type has a v2 gate to apply, namely the
handover $k=4$ and $k=8$ gates for the choice templates and the SLA-breach gate for the yes/no templates; the
congestion-level template, a four-level score, has no v2 class of its answer type and is therefore not gated. The pre-registered outcome counts (safe / within the Proposition~\ref{prop:ts} bound / violating) at
$\alpha=0.05$ were 6/1/0 for Jev,
4/3/0 for AnyJev L0 and
1/6/0 for Laya FT. Laya's losses sit \emph{within} the bound
because its measured shift (0.18--0.90) is as large as the loss itself: this is the
worst case of Proposition~\ref{prop:ts}, and the bound, although correct, requires labels of the new type in order to be
evaluated. Only a declared-type check prevents the failure.

Table~\ref{tab:g3} decomposes the result per template. The fine-tuned specialist gains exactly where the new question's
answer coincides with the trained rule, e.g., under renamed keys (0.193 $\rightarrow$
0.825) and for the RSRP-constrained handover (0.143 $\rightarrow$
0.691), and loses exactly where it inverts the rule, which is the signature of a decision function
of (state, options, family) rather than of the question. The base checkpoint does not read the new questions either,
since its accuracy stays near the chance level on most templates, whereas Jev reached at least 0.98 on six of the eight
templates and AnyJev L0 reached 0.990 on the most-loaded question that defeats the specialist.

\subsection{A5: timeliness}
\begin{figure*}[t]\centering\includegraphics[width=\textwidth]{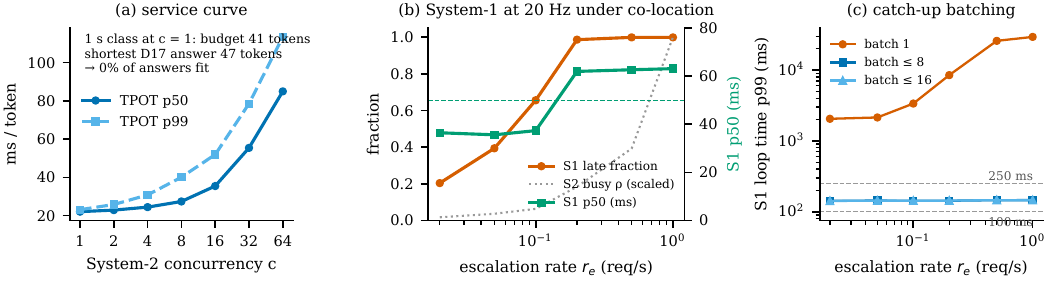}
\caption{Timeliness on the shared L4 (System-1: Laya at $k=8$, 20 decisions/s; System-2: Qwen3-8B-AWQ under vLLM).
(a) System-2 service curve; the ``shortest answer'' annotation refers to the shortest free-text System-2 answer on the
test split (47 tokens). (b) A 20\,Hz System-1 under co-located escalations; the busy-fraction curve
$\rho$ in this panel is the System-2 busy fraction $u_2$ of Proposition~\ref{prop:coloc}. (c) The same loop with catch-up batching of System-1 (loop time = completion $-$ due
time).}\label{fig:time}
\end{figure*}
System-1 timeliness separates the design points before any escalation is considered. The local specialists stay within
about 30\,ms at p99 up to $k=8$ (Table~\ref{tab:fw} lists $k=4$), whereas the zero-label AnyJev L0 grows with $k$, since
it marginalizes over $k$ cyclic option orders. Jev's median latency is flat in $k$, as expected of a network-bound service
whose TCP connection alone takes 199\,ms at the median from the platform's site, and its warm-connection
p50 (281\,ms at $k=4$) already sits above the tightest near-RT classes.

At concurrency 1, the co-located System-2 decodes at 22.1\,ms/token after a 94\,ms first token, so a
1\,s class leaves a budget of 41 tokens, while the shortest System-2 answer was 47
tokens (45 quantized). Free-text escalation to this System-2 at or below 1\,s is therefore infeasible (Lemma~\ref{lem:floor}), and by
Corollary~\ref{cor:collapse} the optimal gate in that regime is act-or-abstain. Co-location, in turn, costs more than the
escalation itself (Fig.~\ref{fig:time}b): while System-2 decodes, each System-1 decision takes 67\,ms at
p99, above the 50\,ms period, so that at one escalation every 50\,s, 20\% of decisions are already late.
Proposition~\ref{prop:coloc} places the stability boundary at a System-2 busy fraction $u^*\approx0.66$ (with
$T=50$\,ms, $s_i=36.2$\,ms and $s_b=62.0$\,ms). Between 0.1 and 0.2 requests/s System-2's
offered load rises from 0.64 to 1.88, and its busy fraction, $u_2\approx1-e^{-\rho}$ under the
light-load approximation, from 0.47 to 0.85, i.e., across $u^*$; the System-1 median switches in the
same interval from near its idle value (37.4\,ms at 0.1 requests/s) to its busy value (62.0\,ms),
as the proposition predicts. The same
proposition indicates the remedy, namely bringing the busy-period cost per decision below $T$. Serving every already-due
decision in one forward pass (catch-up batching, Fig.~\ref{fig:time}c) achieves this: during decode bursts, batches of
2 form at 41.5\,ms per decision (60.6\,ms unbatched), and the loop-time p99 stays at
142--146\,ms at every rate up to 1 request/s, against 2.0--29\,s without
batching. No decision misses a 250\,ms deadline, but 12--56\% miss 100\,ms.

\begin{figure}[t]\centering\includegraphics[width=\linewidth]{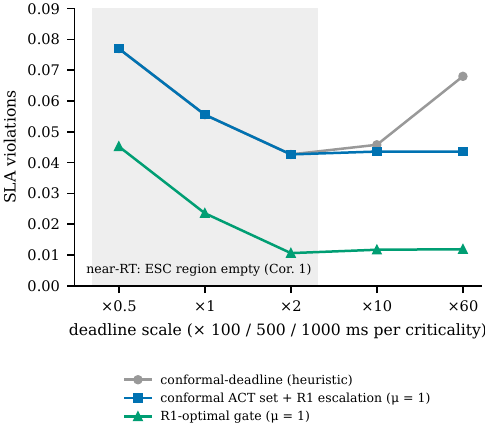}
\caption{Theorem~\ref{thm:gate} in a trace-driven replay: SLA violations vs.\ deadline scale ($\alpha=0.05$). The R1
escalation rule equals the heuristic deadline check wherever escalation is infeasible and improves on it where it opens.
The R1-optimal gate abstains more (coverage 0.629 vs.\ 0.861).}\label{fig:r1}
\end{figure}
Where escalation becomes feasible, Theorem~\ref{thm:gate} determines which decisions to escalate. In a trace-driven replay
with measured System-1 and System-2 latencies and accuracies, we replaced the heuristic rule ``escalate if System-2's p99
fits'' with the rule of Theorem~\ref{thm:gate} (termed R1 hereafter; pricing a miss like an error); this changed nothing at up to $2\times$ the class deadlines,
and at $60\times$ reduced SLA violations from 0.068 to 0.044, escalating only SLA-breach items
(14.2 per 3,645 decisions), which is the one class where System-2 is more accurate than System-1
(Fig.~\ref{fig:r1}). A third gate, termed R1-optimal, applies Theorem~\ref{thm:gate} with a tuned-threshold \ACT{} set; it records the fewest SLA violations at every
deadline scale, but it does so by abstaining more (coverage 0.629 at $60\times$), which is the coverage
side of the same trade-off; the like-for-like comparison is therefore between the two gates that share the conformal \ACT{} set and differ only in
the escalation rule (conformal+R1 and conformal-deadline in Fig.~\ref{fig:r1}). With $\mu=0$, i.e., when a missed deadline is not priced,
the rule also escalated a few handover decisions whose answers then expired, as the first remark on
Theorem~\ref{thm:gate} anticipates.

\subsection{A6: cardinality}
Laya's option budget truncates each option to a few tokens at $k\ge16$, so that direct choice there is at chance level,
and AnyJev labels options with letters, which limits it to $k\le26$. Above the option budget, the map of
Section~\ref{sec:apply} serves the local specialists through the group-and-final reduction of Section~\ref{sec:theory}
(groups of 8 followed by a final), with its measured accuracy, latency and certified gate (Table~\ref{tab:gate}); the
analysis of where this reduction pays off is reported in an extended version of this work.

\subsection{A7: rewording}
Under four rewordings of the known questions, fine-tuned Laya's accuracy moved by at most 0.4\,pts (handover $k\le8$:
0.908 at the original wording), and the pre-registered rewording endpoint held in 4 of 4
families. Given A4, this is insensitivity rather than robustness: an invariance test cannot tell the two apart, whereas a
counterfactual question can. The plain encoder was stable under paraphrase, but fell from 0.948 to
0.263 when the handover rule was written into the question; AnyJev L2's trained heads dropped
24--45\,pts; and Jev \emph{gained} when the rule was stated (anomaly 0.347 $\rightarrow$
0.880). Certified robustness to word edits by randomized smoothing~\cite{cohen,ranmask} did not
separate the models within this budget: over five seeds, 75--84\% of certified-correct decisions
sit at the largest lower bound that 100 samples can give, so the certified radius reflects the sampling budget and not the
model, and we make no radius claim.

\subsection{A8: energy and sovereignty}
In terms of energy efficiency, a $k=4$ decision on the L4 costs 0.55\,J with the encoder, 1.10\,J with
Laya and 19.6\,J with AnyJev-8B L0, i.e., more than an order of magnitude more for the LLM-based
question-reader, and among the LLM-based variants only the smaller AnyJev-1.7B with trained L2 heads comes close to the
specialists, at the price of markedly lower accuracy on the known classes. Jev's energy is not measurable locally, and its
inputs, i.e., network state and KPIs, leave the operator's domain, which turns sovereignty into a deployment criterion
alongside latency and energy. Each framework was also placed behind the same typed endpoint of the core-policy loop of
Fig.~\ref{fig:testbed}, and every chosen policy change reached the SMF and the UPF data path; since those decisions were
zero-shot and have no outcome labels, the loop establishes integration rather than policy quality, and its
per-framework analysis is reported in an extended version of this work.
\begin{figure}[t]\centering\includegraphics[width=\linewidth]{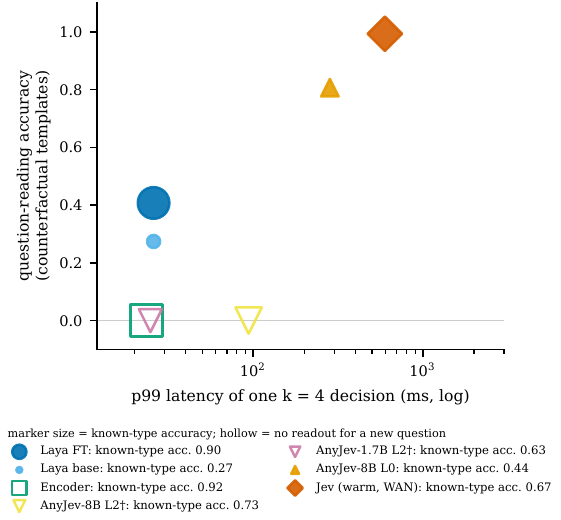}
\caption{The trade-off space: latency, question-reading and known-type accuracy per framework variant.}\label{fig:trade}
\end{figure}

\section{Where Each Framework Applies}
\label{sec:apply}
\begin{figure*}[t]\centering\includegraphics[width=\textwidth]{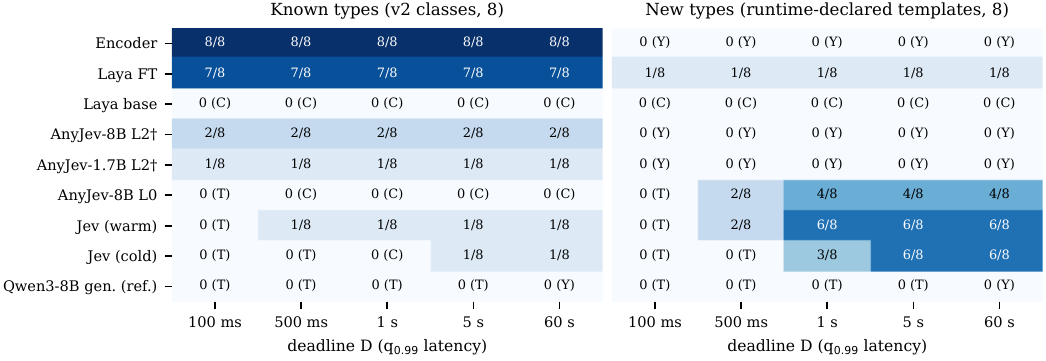}
\caption{Applicability map (Proposition~\ref{prop:A}): per framework variant and deadline, the number of decision classes
for which timeliness (T), type (Y), coverage (C) and cardinality (K) all hold at $\alpha=0.05$, $c_d=0.5$ on a dedicated
GPU; ``0 (P)'' names the most frequent violated predicate. $\dagger$: extension with trained L2 heads.}\label{fig:map}
\end{figure*}
Fig.~\ref{fig:map} evaluates Proposition~\ref{prop:A} for every framework, class and deadline from the measurements alone.
For known types, the fine-tuned specialists are applicable from 100\,ms on a dedicated GPU (encoder
8/8, Laya 7/8); with a decoding System-2 co-located on the same GPU the 100\,ms
column does not hold, since 12--56\% of decisions are late even with catch-up
batching (Section~\ref{sec:results}), and the first applicable class is then about 250\,ms. No zero-shot framework reaches the coverage floor on a known rule-defined class, except Jev on
SLA breach. For new types, the question-reading frameworks are applicable, and only from 0.5--1\,s: Jev
6/8 (warm connection) and AnyJev L0 4/8 at 1\,s. The specialists fail (Y) or (C), with
the exception of fine-tuned Laya on the renamed-keys template (1/8), which is its known question under
new option names. Fig.~\ref{fig:trade} presents the same trade-off as a design space.

\textbf{Per framework.} \emph{Jev} offers question-reading at zero training labels: it answered the counterfactual
templates at 0.99, gained on anomaly severity when the rule was stated in the question, and its gate was safe on six of seven new
templates. Its cost is a network-bound latency that never reached 100\,ms (warm p50 281\,ms, cold
726\,ms), low coverage on rule-defined known classes, and the loss of sovereignty over the decision
inputs. \emph{Laya} offers a fast, local and fine-tunable typed head (p99 25.8\,ms; 1.10\,J), which
reaches specialist accuracy on its trained questions, although
its certified \ACT{} rate at handover $k=64$ (0.43, Table~\ref{tab:gate}) falls below the coverage floor, which makes handover $k=64$
the one known class it does not serve in Fig.~\ref{fig:map}; after fine-tuning it
does not read the question, so it must only serve declared types on which it was calibrated, and a new type costs about
100 labels and a full fine-tuning pass (0.78--0.97 on the new templates). Where its label
efficiency exceeds that of a plain encoder (handover at 100--300 labels), Laya is the natural specialist to stand up
quickly; otherwise a plain encoder performs as well or better, at approximately half the energy per decision.
\emph{AnyJev} offers an open question-reader: L0 read the counterfactual templates at 0.81, but it is slow (p99
873\,ms at $k=8$), and acts rarely at a strict risk budget, while
its trained L2 heads are fast but inherit the layout-bound failure of any per-type head.

\textbf{Trade-off.} The characterization highlights a practical trade-off between speed and generality that no single
design point resolves. The local specialists meet near-RT deadlines at 0.55--1.10\,J per decision but answer only the
questions on which they were trained, whereas the question-reading frameworks generalize to runtime-declared questions
at the price of a latency that excludes the classes below about 500\,ms, an order of magnitude more energy (AnyJev L0) or
data that leave the operator's domain (Jev). Type safety is common to all three and settles none of these costs.

\textbf{Deployment guidance.} The results translate into four rules: (i) dispatch on the declared decision type and not on
confidence, sending known types to a local specialist behind a per-type gate and new types to a question-reading model or
to the safe default; (ii) treat near-RT classes as act-or-abstain, do not provision free-text escalation to an 8B-class on-edge LLM at or
below 1\,s, and never co-locate a decoding System-2 with a batch-1 System-1 that serves a periodic loop; with catch-up batching,
co-location is safe for deadlines of about 250\,ms and above, but not for 100\,ms classes; (iii) budget $\lceil1/\alpha\rceil-1$ labels per new
type even for zero-shot frameworks, and hundreds to a thousand for a specialist; and (iv) weigh hosted question-reading against
sovereignty, since the data leave the domain.

\begin{table}[t]\centering\scriptsize
\caption{Pre-registered endpoints of the analyses reported in this paper and their outcomes (all reported as registered).}\label{tab:endpoints}
\setlength{\tabcolsep}{2.5pt}
\begin{tabular}{@{}p{0.34\columnwidth}p{0.36\columnwidth}p{0.2\columnwidth}@{}}\toprule
Endpoint & Measured & Outcome \\\midrule
Label efficiency (ratio $\ge3$ on $\ge3$ families) & 2 of 4 families & failed \\
Retention on new templates (FT $-$ base) & $+$0.138 [$+$0.114, $+$0.163] & held; question-blind \\
Rewording (Laya drop $\le$3\,pts where AnyJev L2 $\ge$10) & 4 of 4 families & held; insensitivity \\
Type mismatch: question-readers safe or within bound; Laya FT violates & Jev 0, AnyJev 0, Laya FT 0 violations & half held \\
R1 gate escalates only where $\pi_2>\max(m,\beta)$, equals the deadline check at $\le2\times$ & only SLA breach; identical at $\le2\times$ & held \\
1\,s class infeasible at $c{=}1$ (Lemma~\ref{lem:floor}) & 41 tokens vs.\ $\ge$45 & held \\\bottomrule
\end{tabular}
\end{table}

\section{Discussion}
\label{sec:discussion}
\textbf{Typed dispatch as an architectural pattern.} The results argue for moving the notion of type from the output
domain to the decision itself. A holistic control stack can declare, for every decision it issues, the question template
and option schema together with the deadline class, and route the decision on that declaration: to a local specialist
behind a per-type gate when the type is known and calibrated, to a question-reading framework when the type is new and
the deadline allows it, and otherwise to the safe default. The declared type then carries the certificate of
Theorem~\ref{thm:crc}, and Proposition~\ref{prop:ts} explains why no confidence signal can replace it. This pattern is
compatible with agentic O-RAN designs that already exchange typed policy instances between the non-RT and near-RT
tiers~\cite{a1gent}, and with arbitration layers that govern actions after they are proposed~\cite{xtruce,aura}; the
presented characterization supplies the per-type timing, coverage and label figures that such a dispatcher needs.

\textbf{Evaluation methodology.} Two pre-registered generality endpoints passed while the underlying capability was
absent, which carries a lesson beyond the frameworks evaluated here. Invariance tests, i.e., rewording a question and
checking that the answer does not change, cannot distinguish a model that understands the question from one that ignores
it. Counterfactual questions, which keep the state and options and change the question so that the correct answer
changes, can. Characterizations of decision models for network control should therefore include counterfactual
questions, and should report the share of unchanged answers next to accuracy.

\textbf{Energy and sovereignty.} Energy per decision separates the design points by more than an order of magnitude at
equal option count (Section~\ref{sec:results}), and in the presented measurements the most energy efficient choice for known
types, namely the plain encoder, was also the most accurate one. The hosted framework removes the local energy cost from
the operator's accounting but moves the decision inputs, i.e., network state and KPIs, outside the operator's domain.
Both properties belong in the deployment decision next to latency and accuracy, and neither is visible in an
accuracy-only comparison.

\section{Limitations}
\label{sec:limits}
The presented characterization carries the following limitations, which are stated plainly. (1) The pre-registered
label-efficiency claim held in only 2 of 4 families, as summarized in Table~\ref{tab:endpoints}. (2) Fine-tuned Laya is
question-blind, and two pre-registered generality endpoints passed only because of it. (3) Plane A (E2/near-RT RIC) is not
closed, so the near-RT results are measured in isolation and by replay. (4) There is no real RAN, and RAN timing is scoped
by a real-time factor. (5) Labels come from simulator rules. (6) One 8B LLM serves as both System-2 and accuracy ceiling.
(7) The shift split barely moves the state distribution, so the shift results are weak evidence, with the one
confidence-level exception noted in Section~\ref{sec:results}. (8) Jev is hosted and versioned (jev-1.13.0), and its
behaviour may change. (9) The analyses of the group-and-final reduction and of the per-framework core-policy loop,
together with their pre-registered endpoints, are not part of this paper; the map uses only the measured accuracy,
latency and gate of the reduction. (10) Question-contrastive fine-tuning and GPU partitioning (NVIDIA Multi-Process
Service, MPS, not run on the shared edge GPU) were not measured, and the applicability map is computed for a dedicated
GPU. (11) The label-efficiency comparison does not separate the typed head from the backbone size, since Laya's backbone
is ModernBERT-large and the plain encoder's is ModernBERT-base. (12) Certified robustness radii are bounded by the
certification sample budget ($n=100$).

\section{Conclusion}
\label{sec:concl}
This paper presented a theory-driven characterization of type-safe decision frameworks for agentic 5G control, which
aims to establish where each design point can be applied rather than to rank them. The manuscript not only presented a
theoretical framework whose results compose into checkable applicability predicates, but also measured every predicate
for a hosted, an open fine-tunable and a retrofit framework on an Open5GS/UERANSIM platform with a closed core-policy loop
and an NVIDIA L4 edge GPU. The characterization showed that type safety makes agentic control executable, but that safety
of the output type is not safety of the decision: a fine-tuned typed encoder answered its training question for
98--99.5\% of changed questions while its calibrated gate kept acting, whereas the
question-reading frameworks acted wrongly on at most 0.143 and 0.137 of any changed question,
at 11--29 times the latency for the hosted Jev and at an 8B language model's energy cost for
AnyJev. In addition, free-text escalation to the evaluated 8B-class on-edge LLM proved infeasible at or below 1\,s, a
co-located LLM made a batch-1 20\,Hz loop late at any escalation rate unless System-1 batched its due decisions, and
certified action required at least 19 labels per new decision type even for zero-shot frameworks. In contrast to
evaluations that rank models by accuracy, the presented applicability map assigns known decision types to local specialists from 100\,ms on a dedicated GPU
(from about 250\,ms when a decoding System-2 shares it) and runtime-declared types to question-reading frameworks from 0.5--1\,s, subject to energy and
sovereignty constraints. Future work will focus on question-contrastive fine-tuning of typed encoders, on closing the
near-RT loop through the E2 interface with a real RAN, and on GPU partitioning between System-1 and System-2.

\section*{Acknowledgment}
The research leading to these results has been supported by the SHARC project (Grant Agreement No.~101290994) and the
PQ-NEXT project (Grant Agreement No.~101225759).

\bibliographystyle{IEEEtran}
\bibliography{references}

\appendix[Proofs]
\begin{IEEEproof}[Proof of Theorem~\ref{thm:gate}]
The program of Section~\ref{sec:theory} is a linear program over randomized rules; under Slater's condition (some rule
meets both constraints strictly) strong duality holds~\cite{altman}, and the Lagrangian, an expectation over $(m,\sigma)$
of the chosen action's value, is maximized pointwise. Under independence, $V_{\ESC}=F_2(\pi_2-\beta+\mu)-\mu-\kappa$.
If $\pi_2\le m$, then $V_{\ACT}-V_{\ESC}\ge(1-F_2)(m-\beta+\mu)+\kappa\ge0$ on
$\{m\ge\beta\}$, while on $\{m<\beta\}$, $V_{\ACT}<0=V_{\ABS}$.
\end{IEEEproof}

\begin{IEEEproof}[Proof of Corollary~\ref{cor:collapse}]
If $F_2(\sigma)=0$, then $V_{\ESC}=-\mu-\kappa\le0$. Otherwise, under independence,
$V_{\ESC}=F_2(\pi_2-\beta+\mu)-\mu-\kappa\le F_2(\pi_2-\beta)\le\max(\pi_2-\beta,0)\le\max(V_{\ACT},V_{\ABS})$.
\end{IEEEproof}

\begin{IEEEproof}[Proof of Lemma~\ref{lem:floor}]
After the first token, decoding emits one token per step, and the concurrency-1 step time is the shortest the server
attains, so the sum is a lower bound that is attained at concurrency 1.
\end{IEEEproof}

\begin{IEEEproof}[Proof of Proposition~\ref{prop:coloc}]
The server's service rate is $1/s_b$ during System-2's busy time and $1/s_i$ during its idle time, so its long-run
capacity is $u_2/s_b+(1-u_2)/s_i$ decisions per unit time when the alternation is slow relative to $s_b$; the load
condition $1/T<u_2/s_b+(1-u_2)/s_i$ of a single server with modulated service~\cite{loynes} rearranges to
$u_2<s_b(T-s_i)/(T(s_b-s_i))$. During a busy period one decision is served per $s_b$ while one arrives per $T$, so the
backlog grows by $s_b-T$ per period.
\end{IEEEproof}

\begin{IEEEproof}[Proof of Theorem~\ref{thm:crc}]
The loss $L_i(\tau)=\mathbf{1}\{m_i\ge\tau,\text{wrong}_i\}$ lies in $[0,1]$ and is
non-increasing in $\tau$. With $T_n$ the candidate set of $\hat\tau$, let
$\hat\tau'$ be the smallest $\tau\in T_n\cup\{m_{n+1}\}$ with $\frac{1}{n+1}\sum_{i\le n+1}L_i(\tau)\le\alpha$. Since
$L_{n+1}\le1$, $\hat\tau$ is feasible for this problem, so $\hat\tau'\le\hat\tau$ and, by monotonicity,
$L_{n+1}(\hat\tau)\le L_{n+1}(\hat\tau')$. The threshold $\hat\tau'$ is a symmetric function of $n+1$ exchangeable
decisions, hence $\mathbb{E}[L_{n+1}(\hat\tau')]=\mathbb{E}[\frac{1}{n+1}\sum_i L_i(\hat\tau')]\le\alpha$, following
conformal risk control~\cite{crc}. For the shift, condition on the calibration set, which fixes $\hat\tau$; the loss is a
$[0,1]$-valued function of $Z$, so $\mathbb{E}_Q L\le\mathbb{E}_P L+d_{\mathrm{TV}}(P_Z,Q_Z)$, and averaging over the
calibration set returns the marginal bound.
\end{IEEEproof}

\begin{IEEEproof}[Proof of Lemma~\ref{lem:C}]
At $\tau=\infty$ the left side of the criterion is $1/(n+1)$; any finite $\tau$ therefore needs
$1/(n+1)\le\alpha$.
\end{IEEEproof}

\begin{IEEEproof}[Proof of Proposition~\ref{prop:ts}]
The bound is Theorem~\ref{thm:crc} with $Q=P^{t'}$. For tightness, fix the calibration set (hence $\hat\tau_t$), give $m$ the
same law under $t$ and $t'$, let $a=\Prob(m\ge\hat\tau_t)$ be the \ACT{} rate, let the gate of type $t$ have
$\Prob_t(\ACT\wedge\text{wrong})=\alpha_t$ exactly, let every acted decision be wrong under $t'$, and couple correctness on
the non-acted decisions identically under both types. Then the two laws of $Z$ differ only in the mass $a-\alpha_t$ that
moves from acted-and-correct to acted-and-wrong, so $d_{\mathrm{TV}}(P^t_Z,P^{t'}_Z)=a-\alpha_t$, while
$\Prob_{t'}(\ACT\wedge\text{wrong})=a=\alpha_t+d_{\mathrm{TV}}$: the bound is attained.
\end{IEEEproof}

\begin{IEEEproof}[Proof of Proposition~\ref{prop:A}]
Compose Theorem~\ref{thm:crc} (risk), Lemma~\ref{lem:floor} and
Proposition~\ref{prop:coloc} (time) and Lemma~\ref{lem:C} (labels); (K) holds by the definition of $K_F^{\max}$, with the
latency of the reduction entering (T), and (C) is the coverage condition itself.
\end{IEEEproof}
\end{document}